\documentclass[11pt,letterpaper]{article}
\usepackage[margin=1in]{geometry}
\usepackage{amsmath,amssymb,amsfonts,mathrsfs}
\usepackage{bm,enumerate}
\usepackage{caption}
\usepackage{algorithm}
\usepackage[noend]{algpseudocode}
\let\oldnl\nl
\newcommand{\nonl}{\renewcommand{\nl}{\let\nl\oldnl}}
\usepackage[english]{babel}
\usepackage[utf8]{inputenc}
\usepackage{latexsym}
\usepackage{amsthm}
\usepackage{subfig}
\usepackage{float}
\usepackage{stmaryrd}
\usepackage{fancybox}
\usepackage{bigstrut,array,multirow}
\usepackage[dvipsnames]{xcolor}
\usepackage{hyperref}
\usepackage[capitalise]{cleveref}
\usepackage{tikz}
\usepackage{tikz-cd}
\usepackage{varwidth}
\usepackage{todonotes}

\usepackage{mathpazo}

\newcommand{\ket}[1]{\vert #1 \rangle}

\newcommand{\mybar}[1]{\lambda}

\newcommand{\ord}[1]{\textrm{ord}(#1)}

\newcommand{\AM}{\mathrm{AM}}
\newcommand{\coAM}{\mathrm{coAM}}
\newcommand{\MA}{\mathrm{MA}}

\newcommand{\NP}{\mathrm{NP}}

\newcommand{\soc}[1]{\mathrm{Soc}(#1)}

\newcommand{\mingen}{\textrm{MIN-GEN}}

\newtheorem{theorem}{Theorem}

\newtheorem{fact}{Fact}

\newtheorem{remark}{Remark}

\newtheorem{lemma}{Lemma}

\newenvironment{proof-sketch}{\trivlist\item[]\emph{Brief proof sketch}:}%
{\unskip\nobreak\hskip 1em plus 1fil\nobreak$\Box$
\parfillskip=0pt%
\endtrivlist}

\begin{document}

\title{Quantum Algorithms for Minimum Generating Set}
\author{
Bireswar Das$^{1}$ \quad
Udit Kumar$^{1}$ \quad
Kavita Samant$^{2}$ \quad
Dhara Thakkar$^{3}$ \\[1mm]
$^{1}$Indian Institute of Technology Gandhinagar, India\\
\quad
$^{2}$Shiv Nadar Institution of Eminence, Delhi-NCR, India
\quad\\
$^{3}$Nagoya University, Nagoya, Japan \\[1mm]
\texttt{\{bireswar,kumarudit\}@iitgn.ac.in},
\quad
\texttt{ks299@snu.edu.in},\\
\quad
\texttt{dharathakkar754@gmail.com}
}
\date{}
\maketitle
\begin{abstract}
    In this paper, we present a polynomial-time quantum algorithm for computing a minimum-sized generating set of solvable black-box groups. Next, we consider the class $\Gamma_d$ of black-box groups, where every non-abelian composition factor is isomorphic to a subgroup of the symmetric group $S_d$ for a fixed $d$. We design polynomial-time quantum algorithms to compute the direct product decomposition of abelian factor groups and solve the constructive membership problem for factor groups of groups from $\Gamma_d$. With the help of these algorithms, we design a quantum algorithm for computing a chief series of black-box groups from $\Gamma_d$. Using the chief series, we construct a polynomial-time quantum algorithm for computing minimum generating sets of black-box groups from $\Gamma_d$.    
    Finally, we show that the minimum generating set problem for general black-box groups is in $\NP \cap \coAM$. 
\end{abstract}
\section{Introduction}
 The minimum size of a generating set of a finite group $G$ is denoted by $d(G)$. In this paper, we study the problem $\mingen$ of computing $d(G)$ for a given group $G$. The minimum generating set problem is a fundamental problem in group theory and computational group theory \cite{gaschutz1959eulersche, Lucchini1995273,LD,derek,CollinsWeiss}. Given a group $G$ by a set of generators, the input generators may contain redundancy, and MIN-GEN asks for a smallest possible subset that generates $G$. In particular, algorithms working with the smallest generating set may potentially use less time and space. Papadimitriou and Yannakakis \cite{PAPADIMITRIOU1996161} conjectured that the analogous problem of computing a minimum generating set for a quasigroup, given by its Cayley table, is complete for $\NP[\log^2 n]$. In the Cayley table model, Arvind and Tor\'an~\cite{arvind2006complexity} designed a $\textrm{DSPACE}(\log^2 n)$ algorithm for $\mingen$ of groups. Collins, Grochow, Levet, and Wei{\ss} \cite{CollinsWeiss} provided strong evidence that the conjecture by Papadimitriou and Yannakakis does not hold for quasigroups as well as groups.

The mathematical characterization of a minimum generating set for abelian groups follows from their direct product decomposition. In 1959, Gasch\"utz \cite{gaschutz1959eulersche} developed some of the most important tools for studying group generation. Lucchini and Menegazzo~\cite{Lucchini1995273} introduced additional techniques, designing the first significant algorithm to solve $\mingen$ for solvable permutation groups in polynomial time.

In recent years, remarkable progress has been made in the design of polynomial-time algorithms for the minimum generating set problem. In a breakthrough paper, Lucchini and Thakkar~\cite{LD} designed the first polynomial-time algorithm to solve $\mingen$ for groups given by their Cayley tables. Shortly thereafter, Holt and Tracey~\cite{derek} developed a polynomial-time Las Vegas algorithm to solve the problem for permutation groups given by their generating sets.

In this paper, we investigate $\mingen$ within the black-box group model. Black-box groups were introduced by Babai and Szemer\'edi~\cite{bS} to handle matrix groups. This framework serves as a vast generalization of both Cayley table and permutation group representations. In this model, group elements are encoded as strings of a fixed length $n$. 
The model comes with an oracle that takes the encodings of elements $g$ and $h$ and outputs the encoding of the product $gh$. The input group is given by a set of generators.
Babai~\cite{babaibounded} essentially showed that an $\NP$ machine cannot distinguish between a black-box group representing $\mathbb{Z}_p$ and one representing $\mathbb{Z}_p \times \mathbb{Z}_p$.
Since  $\mathbb{Z}_p$ is generated by a single element, and $\mathbb{Z}_p\times \mathbb{Z}_p$ needs two generators, there is no polynomial-time classical algorithm to solve $\mingen$ even for abelian black-box groups. However, a quantum algorithm can easily be designed to compute a minimum generating set of an abelian black-box group using Cheung and Mosca's direct product decomposition algorithm \cite{mosca}.

Quantum algorithms have demonstrated advantages for several problems in the black-box model that are classically intractable. 
Watrous \cite{watrous} developed a polynomial-time quantum algorithm to compute the order of solvable black-box groups, a task classically intractable even for abelian groups.   Ivanyos, Magniez, and Santha \cite{santha} extended these results and developed efficient quantum algorithms for several other fundamental group theoretic tasks for larger classes of black-box groups. In 2000, Watrous \cite{Succinct} showed that group non-membership of a black-box group is in $\mathrm{QMA}$ but not in $\MA$ for some black-box oracle. Recently, Le Gall, Nishimura, and Thakkar~\cite{QCMA} showed that the group order verification problem belongs to $\mathrm{QCMA} \cap \mathrm{coQCMA}$.

\paragraph{Our Contribution.} In this paper, we design a quantum algorithm to solve $\mingen$ for solvable black-box groups (Section \ref{section;Minimum Generating Set of  Solvable Group}). Our algorithm is recursive and, in fact, works for factor groups of solvable black-box groups. As a subroutine for this problem, we design a quantum algorithm to compute minimal normal subgroups of factor groups of solvable groups. 

We next consider $\Gamma_d$, the class of groups whose non-abelian composition factors (see \Cref{sec:prelims}) are subgroups of the symmetric group $S_d$ for some fixed $d$. This class was previously investigated by Ivanyos, Magniez, and Santha \cite{santha}, who designed polynomial-time quantum algorithms to solve several interesting group theoretic tasks\footnote{The algorithms presented in \cite{santha} are more general; rather than requiring $d$ to be constant, they remain efficient even when $d$ is polynomially bounded.}. This same class of groups played a central role in Luks's seminal paper \cite{luks1982isomorphism} on the polynomial-time isomorphism algorithm for bounded-degree graphs. For fixed $d$, we design a polynomial-time quantum algorithm to solve $\mingen$ for groups in $\Gamma_d$.

To solve $\mingen$, we design polynomial-time quantum algorithms to perform several fundamental group theoretic tasks in the factor group of black-box groups in $\Gamma_d$, for fixed $d$. These results might be of independent interest. For example, we design quantum algorithms to compute direct product decompositions of abelian factor groups, constructive membership in factor groups, and several other tasks listed by Babai and Beals \cite{bb} and later by Ivanyos, Magniez, and Santha \cite{santha}. We also design a polynomial-time quantum algorithm to compute a chief series of black-box groups in $\Gamma_d$, for fixed $d$. We note that one of the crucial ingredients in our algorithm is a result of Gasch\"utz \cite{gaschutz1959eulersche}, which allows us to design a recursive algorithm using a minimal normal subgroup. This is why we need a chief series rather than a composition series, as a chief series provides a minimal normal subgroup at each recursive step, whereas a composition series only provides simple factors.

In the decision version of MIN-GEN, the input is an integer $k$ and a black-box group $G$ by a generating set. The task is to decide whether $G$ has a generating set of size at most $k$. We give an $\mathrm{AM}$ protocol to compute a minimum generating set for general black-box groups. This shows that the decision version of  $\mingen$ is in $\NP\cap \coAM$.

\begin{remark}
Even though solvable groups are a subclass of $\Gamma_d$, we present the algorithm for solvable groups separately as they are distinct in several aspects. The general $\Gamma_d$ case demands a larger set of group theoretic and algorithmic machinery, whereas the simpler solvable case serves to intuitively illuminate some of the core ideas.    
\end{remark}

\section{Preliminaries}\label{sec:prelims}
This section presents the necessary preliminaries from group theory and quantum computation that will be used throughout the paper. For additional background, we refer the reader to~\cite{robin,quantum}.

A subgroup $H$ of a group $G$ is called \emph{subnormal} if there exists a series of subgroups $H=H_0 \leq H_1\leq \cdots\leq H_k=G$ such that each $H_i$ is \emph{normal} in $H_{i+1}.$ A subnormal series with strict inclusions, where each factor $G_i/G_{i-1}$ is simple, called a \emph{composition
series}, and its factors $G_i/G_{i-1}$ are called \emph{composition factors}. The \emph{commutator} of two elements $x,y$ in a group $G,$ denoted as $[x,y],$ and is defined as $[x,y]=x^{-1}y^{-1}xy.$ The subgroup generated by the set $\{[x,y]:\,x,y\in G\}$, denoted $G'$, is called the \emph{commutator subgroup} of $G$.  
The \emph{higher commutator subgroups} of a group $G$ are defined inductively by
$G^{(0)}=G; \, G^{(i+1)}=(G^{(i)})',$ where $(G^{(i)})'$ is the commutator subgroup of $G^{(i)}.$ Moreover, each $G^{(i)}$ is a characteristic subgroup of $G$, thus normal in $G$. The series $G=G^{(0)}\trianglerighteq G^{(1)}\trianglerighteq  G^{(2)}\trianglerighteq  \cdots$ is called the \emph{derived series} of $G.$ A group $G$ is \emph{solvable} if its derived series terminates at its trivial subgroup. Moreover, every subgroup and every quotient group of a solvable group is solvable.

The parameter $\nu(G)$ is the smallest natural number $\nu$ such that every non-abelian composition factor of $G$ can be viewed as a subgroup of $S_r$ for $r\le\nu$. For a natural number $d$, let $\Gamma_d$ denote the class of finite groups $G$ such that $\nu(G)\le d$.

A non-trivial normal subgroup $N$ of a group $G$ is called a \emph{minimal normal subgroup} if there exists a normal subgroup $X$ of $G$ such that $1\leq X\leq N,$ then $X=N.$ A minimal normal subgroup of a finite group is a direct product of isomorphic simple groups. The \emph{socle} of a group $G,$ denoted by $\operatorname{Soc}(G),$ is generated by all minimal normal subgroups of $G.$ Moreover, $\operatorname{Soc}(G)$ is a direct product of simple groups.  We note that 
$Soc(\mathbb{Z}_{p_1^{e_1}}\times \ldots\times \mathbb{Z}_{p_k^{e_k}})=\mathbb{Z}_{p_1}\times\ldots\times\mathbb{Z}_{p_k}$.

A \emph{chief series} of a group $G$ is a sequence of subgroups
$1=G_0\triangleleft G_1\triangleleft\cdots\triangleleft G_{u-1}\triangleleft G_u=G$ such that for all $1\leq i\leq u,$ $G_i\triangleleft G$ and each factor $G_{i}/G_{i-1}$ is a minimal normal subgroup of $G/G_{i-1}.$
A non-trivial finite group $G$ always possesses a chief series, and each $G_{i}/G_{i-1}$ is called a \emph{chief factor} of $G$. 
Every chief factor of a finite group is a direct product of
isomorphic simple groups.
\

\begin{theorem} [\cite{lift}] \label{thm: abelian lifting} 
    Let $G$ be a group and let $N = \langle n_1,\ldots, n_\ell \rangle$ be an abelian minimal normal subgroup of G. Suppose $G/N = \langle h_1N, \dots , h_dN\rangle$ and $d(G/N) = d.$ Then $d\le d(G)\le d+1$. Moreover,
\begin{enumerate}
\item[(i)] If $d(G) = d,$ then either $G=\langle h_1,\dots,h_d\rangle$ or there exists $i$ and $j$ where $1\le i\le d$ and  $1\le j\le \ell $ such that $G=\langle h_1,\dots ,h_{i-1},h_in_j,h_{i+1},\dots,h_d \rangle$; 
\item[(ii)] If $d(G) = d+1$, then $G=\langle h_1,\ldots,h_d, x\rangle$ for any  $x\in N{\backslash}\{1\}$.
\end{enumerate}
\end{theorem}

\begin{lemma}[\cite{isaac}, see page 114]\label{lem:higher}
    Let $G$ be a group. Let $G^{(i)}$ be a higher commutator subgroup of $G,$ generated by $S= \{g_1,\cdots,g_k\}$. If $N$ is a normal subgroup of $G,$ then $(G/N)^{(i)} = G^{i}N/N= \langle \{g_iN: g_i\in S\}\rangle$. 
\end{lemma}
\begin{lemma}[\cite{isaac}] \label{issac(subnormal)}Let $T$ be a simple subnormal group of $G$. If $T$ is non-abelian, then the normal closure $\langle T^G\rangle$ of $T$ in $G$ is a minimal normal subgroup. Otherwise, $\langle T^G\rangle$  is a $p$-group, where $|T|=p$.

\end{lemma}
\subsection{Black-Box Groups}
Black-box groups were introduced by Babai and Szemer\'edi \cite{bS}, in which each element of a group $G$ is encoded by a string of fixed length $n$, which is called the \emph{encoding length}. Suppose $\pi: G \rightarrow \{0,1\}^n$ is an encoding of elements of $G$ in binary strings. If $\pi$ is an injective map, then each element of $G$ has a \emph{unique encoding}. In general, the encoding of group elements need not be unique. In the black-box group model, the input group is assumed to be given by a generating set. 

We now define black-box groups in the framework of quantum computing. We always consider a unique encoding $\pi:G\rightarrow \{0,1\}^n$ in the quantum setting. For simplicity, we denote $\pi(g)$ by $g$. We have two quantum oracles operating on $2n$ qubits. The first quantum oracle $U_G$ maps $\ket{g}\ket{h}$ to $\ket{g}\ket{gh}$, and second quantum oracle $U_{G^{-1}}$ maps $\ket{g}\ket{h}$ to $\ket{g}\ket{g^{-1}h}$, where $g,h\in G$. These oracles output an error on an invalid input (for more details,
see \cite{Succinct},\cite{watrous},\cite{mosca}).

\emph{A straight line program (SLP)} is a circuit $C(x_1,\ldots,x_k)$ having multiplication gates of fan-in 2, an inverse gate of fan-in 1, along with an output gate and input gates $x_1,\dots x_k$. The \emph{size} is the number of gates in it.
\begin{theorem} \label{thm:reachability} (i) (Reachability Theorem, Babai and Szemer\'edi \cite{bS}) Let $G$ be a group generated by $x_1,\ldots,x_k$. For every $x\in G$ there is an SLP $C$ of size $(1+\log|G|)^2$ such that $x =C(x_1,\ldots,x_k)$. (ii) (Uniform Sampling, Babai \cite{babai1991local}) ~ For any group oracle, there exists a randomized algorithm that, on input $S$ and $\epsilon>0$, outputs an element of the group $H=\langle S\rangle$ in time polynomial in $n+\log(1/\epsilon)$. Moreover, for every $h\in H$,
$\Pr[\text{the algorithm outputs } h]\in(1/|H|-\epsilon,1/|H|+\epsilon).$

\end{theorem}

\begin{theorem}[\cite{fast}]\label{Thm: Solvable derived} Let $G$ be a black-box group given by a set of generators. Then, the following tasks can be performed in Monte Carlo polynomial time:
(i) construct the commutator subgroup;
(ii) decide whether or not $G$ is solvable or nilpotent. If $G$ is solvable, the derived series can be constructed.
\end{theorem}
\begin{theorem}[Watrous~\cite{watrous}]\label{thm:watrous}
     Suppose that $G$ is a solvable black-box group with unique encoding. Then there exists a polynomial-time quantum algorithm in the input size that computes the order of $G$ and tests membership in $G$. 
\end{theorem}
 In the next theorem, the parameter $\nu(G)$ is defined in the previous subsection.
\begin{theorem}[Ivanyos, Magniez and Santha~\cite{santha}] \label{thm: Ivanyos, Magniez and Santha}
   Let $G$ be a finite black-box group with not necessarily unique encoding. Suppose that there are oracles for (a) computing the order of an element of $G$, and (b) solving the constructive membership problem in elementary abelian subgroups of $G$. Then the following tasks can be solved by quantum algorithms of running time polynomial in the input + $\nu(G)$: (i) constructive membership test in $G$, (ii) compute the order of $G$ and a presentation for $G$, (iii) find generators for the center of $G$, (iv) construct a composition series $1=G_0\triangleleft G_1\triangleleft G_2\triangleleft\cdots
\triangleleft G_r=G$ for $G$, together with nice representation of the composition factors $G_i/G_{i+1}$,
(v) find Sylow subgroups of $G$.
\end{theorem}
\begin{theorem}[\cite{santha}]\label{thm:santha}
 Let $G$ be a black-box group with unique encoding. Then the tasks $ (i) - (v)$ listed in Theorem \ref{thm: Ivanyos, Magniez and Santha} can be solved by a quantum algorithm of running time polynomial in the input + $\nu(G)$.
\end{theorem}

\begin{remark}
    During computations for black-box groups, the generating set may tend to grow. One can use a Monte Carlo algorithm \cite{fast} to ensure its size remains $O(n)$, where $n$ is the encoding length.
\end{remark}

\section{Minimum Generating Set of  Solvable Group}\label{section;Minimum Generating Set of  Solvable Group}
Suppose $G$ is a solvable black-box group given by a generating set $S$, with a unique encoding of length $n$. We present a recursive algorithm to compute a minimum generating set of $G$. In the recursive case, the algorithm must handle a more general setting: given a group $G$ together with a normal subgroup $N$ of $G$, specified by a generating set, it computes a minimum generating set of $G/N$.
In Algorithm \ref{fig:mingen}, the basic idea behind the algorithm is described. 
\begin{algorithm}[H]
\renewcommand{\alglinenumber}[1]{\textbf{Step #1:}}
\caption{MIN-GEN$(G/N)$}\label{fig:mingen}

\begin{algorithmic}[1]
\Statex \textbf{Input:} A generating set $S$ of a solvable group $G$, and a generating set $T$ of a normal subgroup $N$ of $G$.

\Statex \textbf{Output:} A minimal generating set of the quotient group $G/N$.
\Statex
\Statex \textbf{MIN-GEN}$(G/N)$

\hspace{-1.4cm}\textbf{Base Case:} If $G = N$, return $\emptyset$. Else if $G/N$ is simple, return $\{gN\}$ for some $g \in G \setminus N$. 

\hspace{-1.4cm}\textbf{Recursive Step:}
\State Compute a minimal normal subgroup $M/N$ of $G/N$. (Note that $M \ne G $).
\State Recursively compute a minimum generating set of $G/M$. (Note that $(G/N)/(M/N) \cong (G/M)$).
\State Use the minimum generating set of $G/M$ to find a minimum generating set of $G/N$ according to Theorem \ref{thm: abelian lifting}.
\end{algorithmic}
\end{algorithm}

\paragraph{Step 3 (Lifting generators).} We begin by describing Step 3 of Algorithm \ref{fig:mingen} in detail. Suppose the recursive call to the algorithm has returned a generating set $\{h_1M,\ldots, h_dM\}$ of $G/M$. We note that $(h_1N)(M/N),\ldots, (h_dN)(M/N)$ is a minimum generating set of $(G/N)/(M/N)$ as $G/M\cong (G/N)/(M/N)$ via the isomorphism $xM\mapsto  (xN)(M/N)$. Let $M/N =\langle m_1N,\ldots, m_\ell N\rangle$. To apply \cref{thm: abelian lifting}, we check whether $\langle h_1N, \ldots, h_dN\rangle= G/N$. This check can be done using Watrous's \cite{watrous} order-finding algorithm for factor groups of a solvable group. The same algorithm can be used to check if $\langle h_1N,\ldots,h_{i-1}N,h_iNm_jN,\ldots,h_dN\rangle=G/N$ for each $1\le i \le d$ and $1\le j \le \ell$. If any of these checks is true, \cref{thm: abelian lifting} guarantees that the size of the minimum generating set continues to the same as that of $G/M$. The algorithm can return a generating set based on which of the above-mentioned checks is satisfied. If none of the checks are satisfied, then the algorithm can return $\langle h_1N,\ldots,h_dN, xN\rangle$, where $xN\in M/N\setminus \{N\}$.

\paragraph{Base Case.} For the base case $G = N$ of Algorithm \ref{fig:mingen}, subgroup equality can be checked in polynomial time \cite{watrous}. For the trivial group, the minimum generating set is \emph{defined} to be $\emptyset$.  When $G/N$ is simple, it has prime order $p$, which can be verified by computing the order of $G$ and the order of $N$ via Theorem \ref{thm:watrous}.

\paragraph{Step 1.} To compute a minimal normal subgroup of $G/N$, we first find the derived series of  $G/N$. Using Theorem~\ref{Thm: Solvable derived}, we can compute the derived series of $G,$ that is,
\(G = G^{(0)} \trianglerighteq G^{(1)} \trianglerighteq \ldots \trianglerighteq G^{(r)} = 1,\)
where $r\le n$, and obtain a generating set of size $O(n)$ for each $G^{(i)}$ of $G.$

 By using Lemma \ref{lem:higher}, we obtain a generating set of size $O(n)$ for each higher commutator subgroup $(G/N)^{(i)}$ using the generators of $G^{(i)}$. Next, we find  index  $i$ such that $(G/N)^{(i+1)} =1$ but $(G/N)^{(i)} \ne 1$. Equivalently, we need to check if $G^{(i+1)}N=N$. In other words, we need to check if $G^{(i+1)} \subseteq N$, which can be done using Theorem \ref{thm:watrous}. Therefore, the required index $i$ can be found by a linear search in the derived series of $G$ until $G^{i+1} \subseteq N$.

Since $G/N$ is solvable, we observe that $(G/N)^{(i)}$ is an abelian subgroup of $G/N$. Now, to decompose this quotient group $(G/N)^{(i)}$ into direct factors, apply Watrous's algorithm for abelian factor groups (see section 4.2 in \cite{watrous}). 
One can check that Watrous's algorithm gives a direct product decomposition of the form $(G/N)^{(i)}=\langle x_1 N \rangle\times\ldots \times \langle x_k N \rangle,$ where the order of $(x_jN)$ is a prime power $p_j^{r_j}$ and $k\leq n$ (See Appendix~\ref{sec:A1} ). 

Let $p=p_1$. We can easily compute the Sylow-$p$ subgroup of $(G/N)^{(i)}$ by just collecting the cyclic factors having an order which is a power of $p$. Reordering the indices, we may assume without loss of generality that $Syl_p((G/N)^{(i)})=\langle x_1 N \rangle\times\ldots \times \langle x_m N \rangle$ where $m\leq k\leq n$. 

Since the higher commutators of a group are characteristic subgroups and $(G/N)^{(i)}$ is abelian, thus $Syl_p((G/N)^{(i)})$ is normal in $G/N$. Moreover, since the socle of a group is characteristic, $\soc{Syl_p((G/N)^{(i)})}$ is normal in $G/N$. 

 We can see that $\soc{Syl_p((G/N)^{(i)})}= \langle u_1 N \rangle\times\ldots \times \langle u_m N \rangle$ where $u_j=x_j^{p^{r_j-1}}$ if $r_j>1$, otherwise $u_j=x_j$. This can be easily computed. Moreover, the socle is normal. Therefore, it will contain at least one minimal normal subgroup of $G/N$. In the following, we crucially use the fact that the socle is an elementary abelian $p$-group.

 Let $G/N =\langle g_1N, \dots, g_sN \rangle.$ Consider the action of $G/N$ on $\soc{Syl_p((G/N)^{(i)})}$ by conjugation. Since $\soc{Syl_p((G/N)^{(i)})}\unlhd G/N$, this action is well-defined, and each generator $g_i N \in \{g_1 N, \dots, g_s N\}$ and $u_j N \in \{u_1 N, \dots, u_{m} N\}$, we can express each $g_i^{-1} u_j g_i N $ in the generating set of $\soc{Syl_p((G/N)^{(i)})}$ as follows:
\(
g_i^{-1} u_j g_i N = u_1^{a_{j1}^{(i)}} u_2^{a_{j2}^{(i)}} \ldots u_{m}^{a_{jm}^{(i)}} N,\quad \textit{ where } \quad 0\le i \le s \textit{ and } 0\le j\le m
.\)
Here, each $a_{jk}^{(i)}$ can be computed efficiently using the constructive membership algorithm for an elementary abelian factor group, designed by  Ivanyos, Magniez and Santha \cite[proof of Theorem 10]{santha}. Each generator $g_i N$ corresponds to a matrix $A^{(i)}$ in $GL_m(p)$, the group of invertible $m\times m$ matrices over $\mathbb{Z}_p$, whose $(j,\ell)$-th entry is denoted by $a_{j\ell }^{(i)}$.
Thus, we obtain a set of matrices $\{A^{(1)}, A^{(2)}, \dots, A^{(s)}\}$ which generates an associative matrix algebra $\mathcal{A}$ over $\mathbb{Z}_p$. Note that $\soc{Syl_p((G/N)^{(i)})}$ can be viewed naturally as a vector space $V \cong \mathbb{Z}_p^m$ over $\mathbb{Z}_p$. Applying Rónyai's Las Vegas polynomial-time algorithm \cite{rony} (possibly multiple times) to $\mathcal{A}$, we can efficiently compute a common invariant subspace $W \subseteq V$ that is irreducible. We recall that $W$ is called \emph{invariant} if  $A^{(i)}W \subseteq W$ for all $i$; $W$  is called \emph{irreducible} if it has no proper nontrivial invariant subspace. This irreducible subspace corresponds to a minimal normal subgroup $M/N$ of $G/N$. Therefore, we obtain the following result.
\begin{theorem}
    Suppose $G$ is a solvable black-box group with unique encoding and $N\unlhd G$. Then there exists a polynomial-time quantum algorithm to compute a minimum generating set of $G/N$.
\end{theorem}

\section{Algorithms for Groups in $\Gamma_d$}\label{sec:Algorithm in gamma d}
In this section, we design quantum algorithms for performing several group theoretic tasks in factor groups of groups from $\Gamma_d$.


We begin this section by describing a procedure to prepare several copies of the uniform superposition state $\ket{H}=|H|^{-1/2}\sum_{h\in H}\ket{h}$, where $H$ is a subgroup of a black-box group $G$.

\subsection{Constructing \texorpdfstring{$\ket{H}$ for $H \in \Gamma_d$}{ket{H} for H in Gamma d}}\label{sec:superposition over H}

First we compute a composition series of $H$ using Theorem \ref{thm:santha}, say $
{1}=H_0\unlhd \cdots
\unlhd H_m=H,$
where each factor $H_j/H_{j-1}$ is simple. This construction requires the cardinality of each subgroup $H_j$, which can also be computed efficiently using Theorem \ref{thm:santha}.  

For each factor $H_j/H_{j-1}$, if it is abelian, we use the method of Watrous \cite{watrous} to construct copies of the uniform superposition $\ket{H_j}$ from copies of $H_{j-1}$. If $H_j/H_{j-1}$ is non-abelian, we first compute a left transversal of $H_{j-1}$ in $H_j$ using the technique of Furst, Hopcroft, and Luks \cite{luks}. The algorithm given by the authors in~\cite{luks} needs a membership test in $H_{j-1}$ and runs in time proportional to $[H_j:H_{j-1}]^2 .|T_j|$ times the runtime of a membership query to $H_{j-1}$, where $T_j$ is the generating set of $H_j$. Since every non-abelian simple factor $H_j/H_{j-1}$ is isomorphic to a subgroup of $S_d$, we have $|H_j/H_{j-1}|\leq d!$, and there is a polynomial-time quantum membership test in $ H_{j-1}$ by Theorem \ref{thm:santha}. Therefore, we can compute a left transversal of $H_{j-1}$ in $H_j$.

Let $x_1,\ldots,x_r$ be 
a left transversal of $H_{j-1}$ in $H_j$, and the inverse of each $x_i$ is also known. Suppose we have $\ell$ copies of $\ket{H_{j-1}}$ in registers $A_1,\dots,A_\ell$. We convert these $\ell$ copies of $\ket{H_{j-1}}$ into $\ell$ copies of $\ket{H_j}$ as follows. We first prepare $r^{-1/2}\sum_{k=1}^r \ket{k}_I\ket{0}_B\ket{H_{j-1}}_{A_1}$. Since the transversal elements are classically known, a reversible operator $U_f$ maps $\ket{k}_I\ket{0}_B\mapsto\ket{k}_I\ket{x_k}_B$. Applying $U_G$ to $B$ and $A_1$ gives $r^{-1/2}\sum_k \ket{k}_I\ket{x_k}_B\ket{x_kH_{j-1}}_{A_1}$. We then apply $U_f^\dagger$, which restores $B$ to $\ket{0}$ and leaves $r^{-1/2}\sum_k\ket{k}_I\ket{0}_B\ket{x_kH_{j-1}}_{A_1}$.

For each $c$ from $1$ to $r$: load $x_c^{-1}$ into register $B$ using $X$ gates; apply $U_G$, sending $A_1$ to $\ket{x_c^{-1}x_kH_{j-1}}$; test membership in $H_{j-1}$ coherently into a new register $F$; apply a bitwise XOR of the classical bit-string $c$ into register $I$ controlled on $F$; then uncompute the membership test, uncompute $U_G$, and clear $B$ by applying $X$ gates. Since $x_k$ is one of elements in the left transversal, $x_c^{-1}x_k\in H_{j-1}$ if and only if $c=k$. Thus, for each $k$, the controlled bitwise XOR  is performed exactly once (when $c=k$), mapping $\ket{k}_I\mapsto=\ket{k\oplus k}=\ket{0}_I$ while $A_1$ is returned to $\ket{x_kH_{j-1}}$ and $B,F$ to $\ket{0}$. After the loop, the state is $\ket{0}_I\ket{0}_B\ket{0}_F\otimes r^{-1/2}\sum_k\ket{x_kH_{j-1}}_{A_1}=\ket{0}\ket{0}\ket{0}\ket{H_j}_{A_1}$. The same procedure can be applied independently to each of the $\ell$ registers, producing $\ell$ copies of $\ket{H_j}$.

\emph{Error analysis}: Suppose first that all operations are error-free except for the membership oracle. At each level, we are doing $2\ell r$ membership queries, and the total number of membership queries is $O(\ell d!\log(|H|))$ since $r\le d!$. The membership-test error can be reduced to $\epsilon=2^{-\operatorname{poly}(n)}$ by polynomially many repetitions in each level. Since $d$ is a fixed constant, $\ell$ and the number of composition factors are polynomially bounded; the overall error is still negligible.  

\subsection{Algorithms for Factor Groups in \texorpdfstring{$ \Gamma_d$}{Gamma d}}
Let $G\in \Gamma_d$ and let $H$ be a normal subgroup of $G.$ There is an efficient quantum algorithm for constructive membership in an abelian factor group $G/H$, provided $H$ is a polynomial-size group or solvable group due to Ivanyos, Magniez, and Santha \cite[proof of Theorem 10]{santha}. One can check that the proof of their theorem works for $\Gamma_d$ if we can efficiently produce (several copies of) the uniform superposition $\ket{H}=|H|^{-1/2}\sum_{h\in H}\ket{h}$. Therefore, the results in Subsection \ref{sec:superposition over H} imply that we can design an efficient quantum algorithm for constructive membership in the abelian factor group $G/H$ in $\Gamma_d$. In addition, we obtain the following result. 

\begin{theorem} \label{thm: Factor-santha}
    Let $d$ be a fixed constant. Then there is a polynomial-time quantum algorithm that, on input a group $G \in \Gamma_d$ with unique encoding and a normal subgroup $N$ of $G$, can perform each of the tasks listed in Theorem \ref{thm: Ivanyos, Magniez and Santha} for the factor group $G/N$.
\end{theorem}
\begin{proof} To apply Theorem \ref{thm: Ivanyos, Magniez and Santha}, we need to show that the tasks in (a) and (b) can be solved efficiently using a quantum algorithm. As discussed above, we have a quantum algorithm for task (b), that is, constructive membership in the abelian subgroups of \(G/N\). We now show that task (a) can also be solved efficiently by a quantum algorithm. Observe that $G$ has a unique encoding, and by Theorem \ref{thm:santha}, we can compute $|G|$. Suppose $|G|=m= p_1^{\alpha_1}p_2^{\alpha_2}\ldots p_r^{\alpha_r} $ for some distinct primes \(p_i\) and $\alpha_i\ge 0$. To compute the order of \(gN \in G/N\), note that \(\ord{gN}\) divides $|G|$. Moreover, if \(g^{m/p_i}\in N\), then \(\ord{gN}\) divides \(m/p_i\). Starting with $m$ and testing whether $g^{m/p_i}\in N$, we can repeatedly obtain either a smaller multiple of the order or the order. Since $N$ is closed under $\Gamma_d$, and thus membership in $N$ can be tested efficiently using Theorem \ref{thm:santha}, we keep on repeating this process until we obtain the smallest integer $\alpha'$ such that $g^{\alpha'}\in N$. Since each iteration removes one prime factor of $G$, this procedure requires at most $O(\log(|G|)$ iterations with a membership test in $N$. Hence, using a quantum algorithm, we can compute $\ord{gN}$ in polynomial time.

Since we have constructed both the oracles (a) and (b), the theorem is proved. 
We note that this algorithm runs in time polynomial in $d!$ and the input size, where the $d!$ comes from the use of oracle (b). But this is fine as $d$ is a constant. 
\end{proof}

In the following theorem, we show that there is an efficient quantum algorithm to decompose an abelian factor group $G/N$ into a direct product of cyclic groups, where $G\in \Gamma_d$.

\begin{theorem} \label{thm: abelian factor decompose}
Let $d$ be a fixed constant. Let $G$ be a uniquely encoded black-box group in $\Gamma_d$, and let $N$ be a normal subgroup of $G$. If $G/N$ is abelian, then there is a polynomial-time quantum algorithm that decomposes $G/N$ into a direct product of cyclic groups.
\end{theorem} 
\begin{proof} To decompose the abelian quotient group $G/N$, we note that Watrous's approach \cite{watrous} for factoring abelian groups works if we have several copies of $\ket{N}$, in addition to the order of each generator $g \in A$. Because $G$ is uniquely encoded, we can compute the order of any $g$ using Shor's algorithm and efficiently prepare multiple copies of $\ket{N}$ as described in Subsection \ref{sec:superposition over H}. 
\end{proof}

\subsection {Computing a Chief Series in \texorpdfstring{$\Gamma_d$}{Gammad}}
Suppose $G$ is a black-box group with unique encoding in $\Gamma_d$, and let $N\unlhd G$. For a given composition series of $G,$ we can compute a composition series of $G/N$ in the following way. 

By Theorem \ref{thm:santha}, we can compute a composition series of $G$, say ${1}=G_0\unlhd \cdots\unlhd G_m=G$. Applying Theorem \ref{thm: Factor-santha}, and deleting consecutive duplicates from
$1=(G_0N)/{N}\unlhd \cdots \unlhd (G_mN)/N=(G)/{N}$ yields a composition series of \(G/N\). This observation will be used in our recursive algorithm for computing a chief series of \(G\) in \(\Gamma_d\).
Further, if \(i_0\) is the smallest index such that $G_{i_0}N/N \ne 1$, then $G_{i_0}N/N $ is a simple group. Therefore, $H/N$ is a subnormal simple subgroup of $G/N$, where $H = G_{i_0}N$. 
\begin{theorem}
    Let $d$  be a fixed constant and $G$ be a black-box group in $\Gamma_d$ with unique encoding. Then, there exists a polynomial-time quantum algorithm to compute a chief series of $G$.
\end{theorem}
\begin{proof} 
The algorithm for finding a chief series is recursive. In the base case, we first find a subnormal simple subgroup $H$  of $G$ as described above. We then find the normal closure of $H$ in $G$, denoted by $\langle H^G\rangle$, using the algorithm by Babai, Cooperman, Finkelstein, Luks, and Seress \cite{fast}. If $H$ is non-abelian, then $\langle H^G\rangle$ is a minimal normal subgroup of $G$; otherwise $\langle H^G\rangle$ is a $p$-group by Lemma \ref{issac(subnormal)}. We compute a derived series of $G$ using Theorem \ref{Thm: Solvable derived}, and take the highest non-trivial commutator subgroup, say $X$. We know $X$ is an abelian and normal subgroup of $G$. We decompose $X$ into a product of cyclic groups, each of which is a cyclic $ p$-group, using Theorem \ref{thm: abelian factor decompose}. After this, the computation of $\soc{X}$ and a minimal normal subgroup of $G$ follows the same process described in Section \ref{section;Minimum Generating Set of  Solvable Group}, in which a minimal normal subgroup is computed by finding an irreducible subspace of the socle (which can be viewed as a vector space) under the action of a matrix algebra. The generator of the algebra can be computed using the constructive membership testing from  Theorem~\ref{thm: Factor-santha}. 
Note that in Section \ref{section;Minimum Generating Set of  Solvable Group}, we do not use the solvability of $G$ (more precisely, $G/N$) once we obtain the socle.

 Suppose we have inductively computed a normal series \(1=N_0\triangleleft N_1 \triangleleft \cdots \triangleleft N_k\triangleleft G\) such that $N_{i+1}/N_{i}$ is a minimal normal subgroup of $G/N_i$ for $i = 0,\ldots,k-1$. If $G/N_k$ is simple, then we obtain a chief series of $G$. Otherwise, we describe a method to refine the normal series further.  To do this, we compute a minimal normal subgroup $N_{k+1}/N_{k}$ of $G/N_k$. First, we compute a subnormal simple subgroup $H/N_k$ of $G/N_k$ as described earlier, and then compute $\langle (H/N_k)^{G/N_k}\rangle$ using the algorithm of  Babai, Cooperman, Finkelstein, Luks, and Seress \cite{fast}. If $H/N_k$ is non-abelian, then we can set $N_{k+1}/N_k =\langle (H/N_k)^{G/N_k}\rangle$. Otherwise, $\langle (H/N_k)^{G/N_k}\rangle$ is a $p$-group for some prime $p$. We compute a derived series of $\langle (H/N_k)^{G/N_k}\rangle$ using an algorithm of Babai, Cooperman, Finkelstein, Luks, and Seress \cite{fast}. Observe that the highest non-trivial commutator subgroup in the derived series is abelian and of the form $X/N_k$. Therefore, we can decompose $X/N_k$ using Theorem \ref{thm: abelian factor decompose} obtaining say, $X/N_k = \langle x_1N_K\rangle\times \ldots\times\langle x_kN_k\rangle$. Suppose the order of $x_iN_k$ is $p^{r_i}$, which is efficiently computable. Then,  $\soc{X/N_k} = \langle y_1\rangle\times \ldots\times\langle y_k\rangle$, where $y_i = x_i^{p^{r_i-1}}N_k$ if $r_i>1$, otherwise $y_i = x_iN_k$. 

 Note that $\soc{X/N_k}$ is an elementary abelian $p$-group and therefore can be regarded as a vector space $V$. Consider the conjugation action of $G/N_k$ on $\soc{X/N_k}$. Since constructive membership can be performed in factor groups by Theorem \ref{thm: Factor-santha}, we can apply the same procedure as in Section \ref{section;Minimum Generating Set of  Solvable Group} to compute an irreducible subspace of $V$. This irreducible subspace of $V$ corresponds to a minimal normal subgroup $N_{k+1}/N_k$ of $G/N_k$.
\end{proof}

\section{Finding Minimum Generating Set for the Group Class \texorpdfstring{$\Gamma_d$}{Gammad}}


With the subroutines prepared in Section~\ref{sec:Algorithm in gamma d}, we are now ready to describe a quantum algorithm for $\mingen$ for groups in $\Gamma_d$, where $d$ is a fixed constant. The following important result regarding the minimum generating set gives us a clue for finding a minimum generating set of a group $G$ in $\Gamma_d$.

\begin{theorem}[Gasch\"utz~\cite{GaschutzWolfgang}, Lucchini~\cite{Lucchini1995273}]\label{thm:lift}
Let $\frak{N}$ be a minimal normal subgroup of $\mathfrak{G}$ with $d(\frak{G}/\frak{N})=\ell$. Let $g_1,g_2,\ldots,g_{\ell}\in \frak{G}$ be such that $\frak{G}/\frak{N} = \langle g_1\frak{N},\cdots,g_{\ell}\frak{N}\rangle$. Then $\ell\le d(\frak{G})\le \ell+1$. Moreover, if $d(\frak{G}) = \ell,$ then there exist $n_1,\cdots,n_{\ell}\in \frak{N}$ such that $\frak{G}=\langle g_1n_1,\cdots ,g_{\ell}n_{\ell}\rangle$. Otherwise, if $d(\frak{G}) = \ell+1,$ then there exist $n_1,\cdots,n_\ell, n_{\ell+1}\in \frak{N}$ such that $\frak{G}=\langle g_1n_1,\cdots ,g_{\ell}n_\ell, n_{\ell+1}\rangle$.
\end{theorem}

The above theorem gives a strategy to compute a minimum generating set of $\frak{G}$ from a minimum generating set of $\frak{G}/\frak{N}$: Try all possible choices of $(n_1,\ldots, n_d)\in \frak{N}^d$ (or $(n_1,\ldots, n_d,n_{d+1})\in \frak{N}^{d+1}$). However, this process is too expensive. Instead, we use the method of Lucchini and Thakkar \cite[Proposition 9 and the proof of Corollary 13]{LD}. Holt and Tracey \cite[Lemma 7]{derek} proved that the method of Lucchini and Thakkar, called the ``lifting step'', succeeds with high probability. We state Holt and Tracey’s theorem below. Note that we are not using the exact version of Holt and Tracey’s theorem.

\begin{lemma}[Lemma 7, Holt and Tracey~\cite{derek}]\label{lemma: derek and tracey}
Let $1=N_0\triangleleft  N_1\triangleleft  \cdots \triangleleft  N_u=G$
be a chief series of $G$. Suppose that $d(G/N_k)= \ell $ and $G/N_k=\langle g_1N_k,\ldots,g_\ell N_k\rangle$ for some $1\leq k< u-1$.
If $d(G/N_{k-1})=\ell$, then for uniformly random and independent choices 
$n_1,\ldots,n_\ell\in N_k$,\[\Pr\left[G/N_{k-1}=\left\langle g_1n_1N_{k-1},\ldots,g_\ell n_\ell N_{k-1}\right\rangle\right]\geq  53/(90\log |G|).\]
Otherwise, if $d(G/N_k)=\ell+1$, then for independent, uniformly random $n_1, \dots, n_{\ell+1} \in N_k$,
\[ \Pr\left[ G/N_{k-1} =\left\langle g_1n_1N_{k-1},\ldots,g_\ell n_\ell N_{k-1},n_{\ell+1}N_{k-1}\right\rangle\right]\geq 53/(90\log |G|).\]
\end{lemma}
The proof of the above Lemma \ref{lemma: derek and tracey} can be obtained by a slight modification of the proof given by Holt and Tracey \cite{derek}, as described in Appendix~\ref{sec:A2}.

\begin{theorem}
    Let $d$ be a fixed constant. Suppose $G$ is a black-box group with unique encoding in $\Gamma_d$. Then there exists a polynomial-time quantum algorithm to compute a minimum generating set of $G$.
\end{theorem}
\begin{proof}

Let $1=N_0\triangleleft N_1\triangleleft\cdots\triangleleft N_u=G$ be a chief series of a group $G$. We describe a quantum algorithm to compute a minimum generating set of $G/N_i$ starting from $i=u-1$ and going down to $i=0$. When $i=0$, the algorithm outputs a minimum generating set of $G/N_0=G$. 

In the base case $i=u-1$, $G/N_{u-1}$ is simple. If $G/N_{u-1}$ is abelian, then any non-trivial element of $G/N_{u-1}$ generates $G/N_{u-1}$. Otherwise, any two uniformly and randomly chosen elements of $G/N_{u-1}$ generate $G/N_{u-1}$ (see Kantor and Lubotzky~\cite{Kantor}). 

Suppose we have found a generating set $\langle g_1N_i,\ldots, g_\ell N_i\rangle$ of $G/N_i$ for some $1\le i<u-1$. 
At this point, Lemma \ref{lemma: derek and tracey} can be used to find a minimum generating set for $G/N_{i-1}$. The ``lifting step'' by picking random elements from $N_i$ is incorporated in Step $2$ and Step $3$ of  Algorithm~\ref{fig:mingengamma}. Note that since the length of the chief series is at most $\log(|G|)$, the total number of executions of  Step $2$ or Step $3$ is at most $\log(|G|)$. Therefore, repeating Step $2$ and Step $3$ for $O(\log|G|)^2$ times in each of the cases ensures that the algorithm succeeds with high probability.     
\end{proof}

\begin{algorithm}[H]
\renewcommand{\alglinenumber}[1]{\textbf{Step #1:}}
\caption{MIN-GEN$(G)$}\label{fig:mingengamma}
\begin{algorithmic}[1]
\Statex \textbf{Input:} A generating set $\{g_1, g_2, \dots, g_k\}$ of a group $G$.
\Statex \textbf{Output:} A minimal generating set of the group $G$.
\Statex
\Statex \textbf{MIN-GEN}$(G)$
\State Compute a chief series of $G$ using previous section, say, $1=N_0\triangleleft  N_1\triangleleft \cdots\triangleleft  N_u=G$

\hspace{-1.4cm}\textbf{Base case:} $i= u-1$
\begin{itemize}
  
        \item If $G/N_{i}$ is abelian, then return any non-trivial element of $G/N_{i}$;
        \item Otherwise, pick a pair of elements uniformly at random from $G/N_{i}$ and check that it generates $G/N_{i}$; Repeat random sampling to find a valid pair.
\end{itemize}

\hspace{-1.4cm}\textbf{Inductive case:} Suppose for $0\le i<u-1$, $d(G/N_i)= \ell$ and $G/N_i = \langle g_1N_i,\ldots, g_\ell N_i\rangle$. 
\State Choose $\ell$ elements uniformly at random from $N_i$, say, $n_1,\cdots,n_\ell$, and if we find $\ell$ elements satisfying, $|G| = |\langle g_1n_1, \ldots, g_\ell n_\ell, N_{i-1}\rangle|$, return $G/N_{i-1} = \langle g_1n_1N_{i-1},\ldots,g_\ell n_\ell N_{i-1}\rangle$. Otherwise, repeat this step for  $O(\log(|G|)^2$ times.

\State If no such sequence of $\ell$ elements is found in Step 2, choose $\ell+1$ elements uniformly at random from $N_i$, say, $n_1,\ldots,n_\ell, n_{\ell+1}$, and if we found any $\ell+1$ elements satisfying, $|G| = |\langle g_1n_1, \ldots, g_\ell n_\ell,n_{\ell+1}, N_{i-1}\rangle|$, return $G/N_{i-1} = \langle g_1n_1N_{i-1},\ldots,g_\ell n_\ell N_{i-1}, n_{\ell+1}N_{i-1}\rangle$. Otherwise, repeat this step for  $O(\log(|G|))^2$ times.
\end{algorithmic}
\end{algorithm}

\section{Complexity of the Minimum Generating Set for Arbitrary Groups}
In this section, we discuss the computational complexity of the decision version of the $\mingen$ problem for black-box groups.  
\begin{theorem}\label{thm:main}
Let $G$ be a finite black-box group given by a generating set $S$ and let $t$ be a positive integer. The problem of deciding whether
$d(G)\le t$ belongs to $\mathrm{NP}$.
\end{theorem}
The proof of this theorem is given in Appendix~\ref{sec:A3}. Next, we prove that deciding whether given sequences of subgroups $(N_0, \cdots, N_u = G)$ is a chief series of a group $G$ is in $\mathrm{AM}$.
\begin{lemma} \label{lemma:AM_CHIEF}
    Let $G$ be a black-box group, and let $(N_0, \cdots,N_u = G)$ be a sequence of subgroups given by generating sets. Then the problem of deciding whether the given sequence is a chief series of $G$ is in $\mathrm{AM}$.
\end{lemma}
   \begin{proof}
 Merlin provides certificates certifying that $N_i \unlhd G$ and $N_{i-1 }\leq N_{i}$ for all $0\le i\le u$. Arthur verifies these certificates in polynomial time by the Reachability Theorem \ref{thm:reachability}.

Arthur can verify in $\mathrm{AM}$ that each factor $N_i/N_{i-1}$ is a minimal normal subgroup of $G/N_{i-1}$ using Babai's $\mathrm{AM}\cap\mathrm{coAM}$ protocol \cite[Corollary 12.1]{babaibounded}.
For this, Merlin supplies the proof for the minimality of $N_i/N_{i-1}$ in $G/N_{i-1}$ for all $i=1,\ldots, u-1$ at the beginning of the protocol.
\end{proof}
\begin{fact}
Suppose $G/N = \langle g_1N, \ldots, g_rN \rangle$. For any $gN \in G/N$, $gN = \operatorname{SLP}(g_1N, \ldots, g_rN) $ if and only if $g^{-1}\operatorname{SLP}(g_1N, \ldots, g_rN) \in N$. 
\end{fact}

\begin{theorem}
    Let $G$ be a black-box group generated by a set $S = \{y_1, \dots, y_r\}$. The problem of computing a minimum generating set is in $\mathrm{AM}$.
\end{theorem}
\begin{proof}
Since constant-round protocols are in $\AM$, by Lemma~\ref{lemma:AM_CHIEF} we assume, without loss of generality, that a chief series $1=N_{0}\triangleleft N_{1}\triangleleft\cdots\triangleleft N_{u}=G$ is given.
    
Merlin supplies a minimum generating set for each $G/N_k$ and also provides certificates for these claims. We describe the certificates subsequently. The protocol proceeds by verifying that Merlin has provided a correct minimum generating set for each $G/N_k$ in parallel.

\vspace{0.2cm}
\noindent   
{\bf For $\bm{G/N_{u-1}}:$} From the properties of chief series, the quotient $G/N_{u-1}$ is simple. If $G/N_{u-1}$ is abelian, Merlin is also supposed to provide a generator $xN_{u-1}$ of $G/N_{u-1}$. Merlin must provide an $\text{SLP}_i(xN_{u-1})$ outputting $y_iN_{u-1}$ for each generator $y_i$.
Similarly, if $G/N_{u-1}$ is non-abelian, then it is generated by two elements (see~\cite{Kantor}); and Merlin's proof must include a generating pair $x_1N_{u-1},x_2N_{u-1}$ for $G/N_{u-1}$, and an $\text{SLP}_i(x_1N_{u-1},x_2N_{u-1})$ outputting $y_iN_{u-1}$ for each $i$.

\vspace{0.2cm}
\noindent
\textbf{For $\bm{G/N_{k-1}}$, $\bm{k\leq u}$:} Suppose that $d(G/N_k)=d$ and  $G/N_k=\langle g_1N_k,\ldots,g_dN_k\rangle$, which is verified by the $k$th parallel run of the protocol. We verify that the generating set provided by Merlin for $G/N_{k-1}$ is indeed its minimum generating set. Here, we have two cases.

\textbf{Case 1.} Suppose Merlin claims that $d(G/N_{k-1})=d$ and sends a generating set $\{x_1N_{k-1},\ldots,x_dN_{k-1}\}$ of $G/N_{k-1}$. Merlin also provide an $\text{SLP}_i(x_1N_{k-1},\ldots,x_dN_{k-1})$ outputting $y_iN_{k-1}$ for each $i \in \{1, \dots, r\}$. Arthur verifies these SLPs; if all SLPs are correct, and we know $d(G/N_k)=d$, then Merlin's claim is correct.

\textbf{Case 2.} Suppose Merlin claims that $d(G/N_{k-1}) = d+1$ and sends a generating set $\{x_1N_{k-1},\ldots,x_{d+1}N_{k-1}\}$ of $G/N_{k-1}$. Merlin also provide an $\text{SLP}_i(x_1N_{k-1},\ldots,x_{d+1}N_{k-1})$ outputting $y_iN_{k-1}$ for each $i \in \{1, \dots, r\}$. Arthur also needs to verify that $d(G/N_{k-1}) \neq d$. To do this, Arthur picks $O(\log |G|)^2$ uniformly random sequences of $d$ elements $(n_1,\ldots,n_d)\in N_k^d$,  send these tuples to Merlin and asks Merlin to prove that $\langle g_1n_1N_{k-1},\ldots,g_dn_dN_{k-1}\rangle \ne G/N_{k-1}$ for each tuple in parallel. This reduces to the $\mathrm{AM}$ protocol for the order verification of a group. This introduces an extra round in the protocol. Arthur rejects the input if any of the $O(\log|G|)$ parallel runs fail.

This concludes the description of the protocol.

\noindent\textbf{Completeness} Suppose that Merlin provides the accurate chief series and the minimal generating sets corresponding to each level of the chief series. By Lemma \ref{lemma:AM_CHIEF}, Arthur will accept the chief series. For $G/N_{u-1}$ and in Case $1$, Merlin provides the generating sets with SLPs. Arthur checks the correctness of SLPs and accepts with probability 1. In Case $2$, no $d$ elements will generates $G/N_{K-1}$, that is,
 $\langle g_1n_1N_{k-1},\ldots,g_dn_dN_{k-1}\rangle \ne G/N_{k-1}$ for any $n_1,\ldots, n_d \in N^d_k$. Therefore, Arthur will accept this with high probability.

 \noindent\textbf{Soundness} Let $k$ be the largest index where Merlin is trying to cheat by claiming
the $d(G/N_{k-1})$ equals $d+1$ while it is actually $d$. Let $p $ be the probability that a random sequence of $d$ elements $n_1,\ldots,n_d\in N_k$ generates $G/N_{k-1}$. By Theorem \ref{lemma: derek and tracey},
\[p=\Pr\left[G/N_{k-1}=\left\langle g_1n_1N_{k-1},\ldots,g_d n_d N_{k-1}\right\rangle\right]\geq\frac{53}{90\log |G|}.\]
Arthur selects $c(\log|G|)^2$ (where $c$ is a constant) uniformly and independently random sequences of $d$ elements from $N_k$. The probability that Arthur fails to pick correct $d$ elements from $c(\log|G|)^2$ that generates $G/N_{k-1}$ is 
$$\Pr[\text{Failure}]= (1-p)^{c(\log|G|)^2)}\le \left(1-\frac{53}{90\log|G|}\right )^{c(\log|G|)^2}.$$
Using the inequality $1-x \le e^{-x}$, $\Pr[\text{Failure}] \le e^{-\frac{53c}{90}\log|G|}$.
This failure probability is exponentially small. Therefore, Arthur rejects Merlin's proof with high probability. Since all verification runs in parallel and the number of interactions
is also constant, the protocol is in $\mathrm{AM}$.
\end{proof}

\bibliographystyle{plain}
\bibliography{main}

\appendix
\section{Appendix}\label{sec:appendix}

\subsection{Generating set for direct product decomposition}\label{sec:A1}
In the following discussion, a procedure for finding a basis of the direct decomposition of $(G/N)^{i}$ by Watrous~\cite {watrous} is explained.

\begin{proof}[\bf{Part of the \textsc{MIN-GEN} Algorithm (Section~\ref{section;Minimum Generating Set of  Solvable Group}})]

We continue with the notation introduced in the preceding discussion. This part is included for completeness. In this remark, we explain how Watrous’s algorithm for factoring abelian groups can be used to compute a generating set for the direct product decomposition of $(G/N)^{(i)}$. Suppose $G^{(i)} = \langle y_1, y_2, \dots, y_t \rangle$ and $L = \operatorname{lcm}(\operatorname{ord}(y_1), \operatorname{ord}(y_2), \dots, \operatorname{ord}(y_t))$. We can efficiently compute the order of each $y_j$ using Shor's order-finding algorithm. Define a mapping
\[
\begin{aligned}
\phi : \mathbb{Z}_L^t &\longrightarrow (G/N)^{(i)} \\
\phi(a_1, a_2, \dots, a_t) &\longmapsto y_1^{a_1} y_2^{a_2} \ldots y_t^{a_t} N
\end{aligned}
\]

Clearly, $\phi$ is a homomorphism with kernel
\[
\ker(\phi) = \{ (a_1, \dots, a_t) \in \mathbb{Z}_L^t \mid y_1^{a_1} y_2^{a_2} \ldots y_t^{a_t} \in N \}
\]
and 
\[
\ker(\phi)^\perp = \left\{ (b_1, \dots, b_t) \in \mathbb{Z}_L^t\ \middle|\ \sum_{j=1}^t a_j b_j \equiv 0 \pmod{L} \quad \forall (a_1, \dots, a_t) \in \ker{(\phi)} \right\}
.\]
\end{proof}
Note that $\mathbb{Z}_L^t/\ker(\phi) \cong (G/N)^{(i)}$. By running Watrous's algorithm for a factor-abelian group \cite{watrous}, we obtain a generating set $\mathcal{X}$ for $\ker(\phi)^\perp$. Suppose $M$ is a matrix whose columns are elements of $\mathcal{X}$. Then $\ker(\phi)$ can be computed, since it is generated by the vectors $a=(a_1,\ldots,a_t)\in \mathbb{Z}_L^t$ satisfying $M^Ta^T =0 $ modulo $L$. We next apply a standard finite abelian group decomposition procedure (see \cite{mosca}). This procedure gives a set of generators $x_1,\ldots , x_r\in \mathbb{Z}_L^t/\ker(\phi)$ such that $\mathbb{Z}_L^t/\ker(\phi)= \langle x_1\rangle\oplus \cdots\oplus\langle x_r\rangle$. Applying $\phi$ to each $x_i$ gives the generating set of $(G/N)^{(i)}$, as desired.

\subsection{Proof sketch for a version of the Lemma given by Holt and Tracey}\label{sec:A2}
We now indicate how the proof presented by Holt and Tracey \cite{derek} can be modified to prove Lemma~\ref{lemma: derek and tracey}. 
A group $G$ acts on a group $X$, called a $G$-set, if there exists a homomorphism from $G$ to $\operatorname{Aut}(X),$ or equivalently, we say $G$ acts on $X$ via automorphisms. A $G$-group $X$ is said to be \emph{irreducible} if $X$ has no proper non-trivial $G$-invariant subgroup. 
We say two $G$-groups $X_1$ and $X_2$ are \emph{$G$-isomorphic}, and write $X_1\cong_G X_2$, if there exists an isomorphism $\phi:X_1\to X_2$ such that
$$(\phi(x))^g=\phi(x^g)\qquad \text{for all } x\in X_1 \text{ and } g\in G.
$$

Two $G$-groups $X_1$ and $X_2$ are said to be \emph{$G$-equivalent},
$X_1\sim_G X_2$, if there are isomorphisms
$\phi:X_1\to X_2$ and $\Phi:X_1\rtimes G\to X_2\rtimes G$
such that the following diagram commutes:
$$
\begin{tikzcd}[column sep=large, row sep=large]
1 \arrow[r] & X_1 \arrow[r] \arrow[d, "\phi"] & X_1\rtimes G \arrow[r] \arrow[d, "\Phi"] & G \arrow[r] \arrow[d, "id"] & 1 \\
1 \arrow[r] & X_2 \arrow[r] & X_2\rtimes G \arrow[r] & G \arrow[r] & 1
\end{tikzcd}
$$
 For a non-abelian chief factor $X,$ the total number of chief factors of $G$ that are $G$-equivalent to $X$ is denoted by $\delta_G(X).$ Note that $\delta_G(X)\leq \log|G|.$ 
\begin{proof}[\bf{Brief proof-sketch of Theorem~\ref{lemma: derek and tracey}}]
 Let $N=N_k/N_{k-1}.$ If $d(G/N_{k-1})=\ell$, then by  Theorem~\ref{thm:lift} there exists $n_1,n_2,\dots,n_\ell \in N_k$ such that $G/N_{k-1}=\langle g_1n_1N_{k-1},\ldots,g_\ell n_\ell N_{k-1}\rangle.$ In the proof of~\cite[Lemma 7]{derek}, Holt and Tracey showed that if $n_1,n_2,\dots,n_\ell \in N_k$ elements are chosen independently and uniformly at random, the probability that  $g_1n_1N_{k-1},\ldots,g_\ell n_\ell N_{k-1}$  generates $G/N_{k-1}$  is $53/(90\delta_{G/N_{k-1}}(N)).$   
 Since $\delta_{G/N_{k-1}}(N)\leq \log|G/N_{k-1}|\leq \log|G|$, we can conclude that \[\Pr\left[G/N_{k-1}=\left\langle g_1n_1N_{k-1},\ldots,g_\ell n_\ell N_{k-1}\right\rangle\right]\geq\frac{53}{90\delta_{G/N_{k-1}}(N)}\geq \frac{53}{90 \log|G|}.\] The proof for the case when $d(G/N_{k-1})=\ell+1$ is similar.
\end{proof}

\subsection{\texorpdfstring{$\mingen$ is in $\NP$}{mingen is in NP}}\label{sec:A3}
\begin{proof}[\bf{Proof of Theorem~\ref{thm:main}}] To prove that $d(G)\le t$, prover provides a generating set $T$. Prover also provide certificates certifying  \(g=\text{SLP}(T)\) for all $g\in S$ and \(h=\text{SLP}(S)\) for all $h\in T$.

The verifier first checks that $|T|\le t$, then verifies each SLP certificate. By the Reachability Theorem, each membership certificate has polynomial length and can be verified in polynomial time. If all certificates are valid, then $G=\langle T\rangle$, and $T$ is a generating set of $G$ of size at most $t$. 
\end{proof}

\end{document}